\documentclass[12pt]{article}

\usepackage[dvipsnames]{xcolor}
\usepackage{amsmath,amsthm,amssymb,natbib,verbatim,cancel, bbm, tikz, tkz-euclide,graphicx,subcaption, float,titling}
\usepackage{natbib}
\usepackage{epigraph}
\usepackage{cleveref}
\usepackage{xfrac}
\usepackage{mathtools}

\theoremstyle{definition}
\newtheorem{theorem}{Theorem}[section]
\newtheorem{definition}[theorem]{Definition}

\newtheorem{example}[theorem]{Example}

\newcommand{\R}{\mathbb{R}}

\usepackage[margin=1in]{geometry} 
\usepackage{setspace}

\title{In Search of Lost Correlation: Correlated Equilibrium via Marginal Actions}
\author{Christopher P. Chambers\thanks{Department of Economics, Georgetown University}\\ Maxime  Cugnon de Sévricourt\thanks{Department of Economics, Georgetown University}\\ Christopher Turansick\thanks{Department of Decision Sciences, Bocconi University}}

\begin{document}

\maketitle
\begin{abstract}
    In this paper, we study which data can be induced by a correlated equilibrium given a known finite simultaneous move game. We assume that an analyst has access to the frequency of each agent's actions but does not have access to the distribution over joint action profiles. We characterize which sets of marginal distributions over actions arise from some correlated equilibria via a type of no arbitrage condition. An outside observer is unable to make a profit in expectation by independently contracting with each agent and collecting a portion of the total utility gained via unilateral deviation. This characterization naturally extends to Nash equilibria.
\end{abstract}


\onehalfspacing
\section{Introduction}

Correlated equilibrium is a powerful tool which allows modelers to capture collusion \citep{dutta2018collusion}, Bayesian rationality given private information \citep{aumann1987correlated}, and pre-play communication \citep{forges2012correlated}. In order to impose antitrust law, tests for collusion have been developed in various environments \citep{bajari2003deciding,aryal2013testing,baranek2025detection}. Further, testing for rationality is the fundamental goal of revealed preference and there have been many tests developed in order to verify consistency with Bayesian rationality \citep{chambers2023coherent,chan2025axiomatic}. There is also a long literature experimentally analyzing pre-play communication \citep{crawford1998survey}. In this paper, we develop a test for correlated equilibria, a common foundation for collusion, Bayesian rationality under uncertainty, and pre-play communication. 

Under full observability of the underlying game and the distribution of the players' joint actions, one can directly test if play constitutes a correlated equilibrium. However, we consider a situation where an analyst observes the frequency with which each player plays an action, but not the joint frequency. As one possible example, suppose that an analyst is trying to test for collusion in a series of government-regulated procurement auctions. The analyst could have access to the history of each firm's bids, but not the joint history, which might be kept confidential. Alternatively, suppose that an analyst is aiming to understand the impact of pre-election (i.e. pre-play) communication on voter turnout. The analyst assumes that agents are strategic in their decision whether or not to vote. In this case, the analyst has access to surveys regarding each individual voter's turnout probabilities, but these surveys are conducted individual by individual and thus do not capture the correlation between one voter and their friends or family.

Our main result shows that the data we consider is consistent with play by some correlated equilibrium if and only if it is immune to exploitation by an outside observer. More precisely, interpreting the payoffs of the game as monetary payoffs, we imagine that an outside observer has access to player's marginals but not the joint distributions. She can offer recommendations to the players, that is she can suggest an alternative distribution over actions, 
but each recommendation comes with a fee. If the observer can make a profit in expectation from such fees, then the strategy profile under consideration is not an equilibrium. This characterization naturally extends to the case of Nash equilibria as well. The main difference between the two characterizations is that, in the case of Nash equilibrium, the recommender is able to condition their fee on the joint realization of action profiles while, in the case of correlated equilibrium, the recommender can only condition their fee on the marginal realization of action profiles. This is a direct consequence of Nash equilibria being defined via the independent mixture of marginal action profiles. Indeed, in the Nash equilibrium case, the recommender only needs to secure a positive expected profit with respect to the joint distribution over actions. In contrast, in the correlated equilibrium case, the recommender must guarantee a profit against every joint distribution over actions that is consistent with players' marginals.

The paper is organized as follows. In \Cref{sec:literature}, we review the relevant literature. In \Cref{sec:model}, we describe the model and delve into the characterization of correlated equilibrium in \Cref{sec:correlated equilibrium} and Nash equilibrium in \Cref{sec: Nash equilibrium}. We conclude in \Cref{sec:conclusion}.

\subsection{Related Literature} \label{sec:literature}
Our paper is related to three strands of literature. The first is the strand which studies correlated equilibria.  The concept of a correlated equilibrium was first introduced in \citet{Aumann1974} and further studied in \citet{aumann1987correlated}.  Since its introduction, many papers have studied various properties of correlated equilibria including its foundation as a result of an evolutionary process \citep{hart2000simple}, its foundation as a result of pre-play communication \citep{forges2012correlated}, as well as the computational costs of computing correlated equilibria \citep{papadimitriou2008computing}. Of particular interest to us are those papers which study the value of correlated equilibria. This includes \citet{ashlagi2008value}, \citet{bradonjic2014price}, and \citet{rudov2025extreme} which compare the value of correlated equilibria, under various objectives, to the value of Nash equilibria.

Our paper also contributes to the strand of literature which studies the empirical content of various solution concepts in strategic settings. In the context of finite normal form games, \citet{sprumont2000testable} and \citet{haile2008empirical} study the empirical content of Nash equilibria and quantal response equilibria, respectively, within the context of a single game. \citet{brandl2024axiomatic} studies the empirical content of Nash equilibria across several games while \citet{sandomirskiy2025monotonicity} studies the empirical content of both Nash equilibria and quantal response equilibria in the context of play across several games. Moving to the setting of sequential games, \citet{bossert2013every} and \citet{rehbeck2014every} find that every choice function and choice correspondence is rationalizable by backwards induction.

Finally, our paper also contributes to the literature which studies the decision theoretic foundations of play in finite simultaneous move games. Within this literature, the most closely related to our own is the work of \citet{Nau1990}. The authors find that a strategy profile is played with positive probability in some correlated equilibrium if and only if it satisfies a series of joint coherency axioms. These joint coherency axioms are posed in terms of bets and arbitrage opportunities by an outside observed. Our main characterization instead focuses on joint profiles of marginal distributions over actions and characterizes when these arise from a correlated equilibrium. Our main axiom is stated in terms of an outside observer (not) making profit by prescribing unilateral deviations to each player in the game. Beyond the work of \citet{Nau1990}, \citet{bernheim1984rationalizable} introduces and studies the axiomatic foundations of rationalizable strategy profiles. \citet{brandenburger1987rationalizability} extends this to analyze a type of correlated rationalizability which ends up being more general concept than joint coherency of \citet{Nau1990}.

\section{The Model} \label{sec:model}

Let $N$ be a finite set of players and for each $i\in N$, $A_i$ is a finite set of strategies or actions corresponding to player $i$. We let $A = \prod_{N}A_i$, with typical element $a$, denote the set of \textbf{action profiles}. Given an action profile $a$, we use $a_i$ to denote the action of agent $i$ and $a_{-i}$ to denote the vector of actions by each player which is not $i$. Each player or agent has a utility function over action profiles $u_i:A\rightarrow\mathbb{R}$.

Given a finite set $B$, we let $\Delta(B)$ denote the set of probability distributions over $B$. A \textbf{marginal strategy profile} $p$ is a list $(p_1,\ldots,p_n)$, where $p_i\in\Delta(A_i)$ for all $i\in N$. A \textbf{joint strategy profile} is some joint distribution $q \in \Delta(A)$. For a joint strategy profile $q$, we use $p_i=\mbox{marg}_{A_i}q$, to denote the marginal distribution of $q$ over $A_i$.

\subsection{Correlated Equilibrium} \label{sec:correlated equilibrium}

Recall the definition of correlated equilibrium \citep{Aumann1974}. 

\begin{definition}
    A distribution over action profiles $q \in \Delta (A)$ is a \textbf{correlated equilibrium} if for all $i\in N$ and all $a_i,a_i'\in A_i$: $$\sum_{a_{-i}}q(a_i,a_{-i})\left [u_i(a_i,a_{-i})-u_i(a'_i,a_{-i}) \right ]\geq 0.$$ 
\end{definition}

In words, a correlated equilibrium is a distribution over action profiles such that, conditional upon learning that she is to play $a_i$, player $i$ has no profitable deviation. A common foundation for this solution concept goes as follows. A mediator selects an action profile $a$ at random according to a joint strategy profile $q$, and then suggests to each player $i$ to take action $a_i$. If $q$ is correlated equilibrium, each player is willing to follow the mediator's suggestion, and therefore $a$ is played.

We assume that the analyst does not observe the full distribution over action profiles $q$, but she observes the marginals over player's actions, that is, she observes a marginal strategy profile $p$. In addition, the analyst observes the utility functions of all the players. This is equivalent to the analyst knowing the underlying game driving players' actions. The analyst is interested in determining whether the observed marginal strategy profile is compatible with the players of the game playing some correlated equilibrium. 
\begin{definition}
A marginal strategy profile $p = (p_1, p_2, \dots, p_N)$ is said to be \textbf{compatible with correlated equilibrium $q \in \Delta(A)$} if:
\begin{enumerate}
    \item for each $i$: $p_i = \mathrm{marg}_{A_i}(q)$,
    \item $q$ is a correlated equilibrium of the game.
\end{enumerate}
Furthermore, we say that $p$ is \textbf{compatible with a correlated equilibrium} if there exists a correlated equilibrium $q \in \Delta(A)$ such that $p$ is compatible with $q$. 
\end{definition}

Now consider the perspective of an external observer who knows, for each player $i$,
the marginal distribution $p_i$ and the payoff function $u_i$, but does not know
the underlying joint distribution over action profiles. The observer seeks to make a profit from the players based on her information subject to each player maximizing their utility. To each player $i$, she proposes a set of fees $f_i : A_i \to \mathbb R$ and a recommendation $\eta_i: A_i \to \Delta (A_i)$. The idea is as follows. Conditional on receiving recommended action $a_i$ from the mediator, the player must pay $f_i(a_i)$ to the observer. In exchange, the player gets recommended a new strategy, $\eta_i(a_i, \cdot) \in \Delta (A_i)$ such the new strategy constitutes a profitable deviation, net of the fee to the observer. That is, for all $a_i$ and $a_{-i}$:
\begin{equation*}
    u_i (a_i, a_{-i}) \leq \sum_{\tilde a_i \in A_i} \eta_i(a_i, \tilde a_i) u_i(\tilde a_i, a_{-i}) - f_i(a_i).
\end{equation*}

This inequality states that that player $i$ profitably deviates by playing the observer's recommendation $\eta_i$ and paying the fee to the observer. 

Under the assumption that each player's utility is quasi-linear in money, it is natural to assume that players can transfer utility among themselves via monetary transfers.\footnote{Alternatively, we could assume that each player's utility function corresponds to a monetary payoff or that players have perfectly transferable utility.} We assume that, for each action profile $a_i$, players can agree \textit{ex ante} to a set of transfers $h_i(a)$ such that $\sum_i h_i(a) =0$. If such transfers are feasible, the observer can profitably induce a deviation from $a=\left(a_i, a_{-i}\right)$ only when the total payoff at $a$ is no greater than the total payoff the observer can generate through the recommended deviations and associated transfers. Formally, this requires
\begin{equation} \label{eq:IC-sum}
    \sum_i u_i\left(a_i, a_{-i}\right) \leq \sum_i \sum_{\tilde{a}_i \in A_i} \eta_i\left({a}_i, \tilde{a}_i\right) u_i\left(\tilde{a}_i, a_{-i}\right)-\sum_i f_i\left(a_i\right).
\end{equation}

Here, the right-hand side represents the aggregate expected payoff the observer can generate by recommending alternative actions $\tilde{a}_i$ according to $\eta_i$, net of the transfers $f_i\left(a_i\right)$ that must be paid when players receive recommendation $a_i$. The observer seeks to make positive expected profits. This is true if and only if the following inequality holds. 
\begin{equation*}
    \sum_i \sum_{a_i} p_i(a_i) f_i(a_i) > 0
\end{equation*}

We say that a marginal strategy profile is \textbf{unexploitable by action-wise transfer schemes} if, for all possible deviations $\eta_i$ and transfers $f_i$ satisfying \eqref{eq:IC-sum}, the observer would make non-positive profits, that is $\sum_i \sum_{a_i} p_i(a_i) f_i(a_i) \leq 0$. We are now ready to state our main theorem.

\begin{theorem} \label{th:correlated_eq}
    A marginal strategy profile $p$ is compatible with a correlated equilibrium if, and only if, it is unexploitable by action-wise transfer schemes. 
\end{theorem}

All the proofs are relegated to the appendix. The core of the argument relies on an application of Farkas' Lemma \citep{borderAlternative} to the system of inequalities capturing non-negativity of $p$, consistency between $p$ and $q$ and the correlated equilibrium incentives constraint. 

\Cref{th:correlated_eq} shows that, taking the marginals as given, there is a link between correlated equilibrium and a notion of rationality in a corresponding transferrable utility game. In particular, if, given a marginal strategy profile, an observer can recommend profitable deviations to players given the action the mediator suggested them to take, and derive a profit from it in expectations, then the strategy profile cannot be a correlated equilibrium. 

\begin{example}\label{ex:coordination}[Exploitation of Strictly Dominated Strategies]\\
Consider the following payoff matrix for a two person simultaneous move game.
\begin{equation}\label{eq:coordination}
    \scalebox{1.3}{$
\begin{array}{c|ccc}
  & L & M & R \\
\hline
T & 9,9 & 0,0 & 0,0 \\
B & 0,0 & 1,1 & 0,0
\end{array}
$}
\end{equation}

\noindent Clearly, action $R$ cannot appear in any correlated equilibrium, as it is strictly dominated by the mixed action $\tfrac12 L + \tfrac12 M$. Suppose, for contradiction, that a marginal strategy profile $p$ assigns positive probability to $R$, i.e.\ $p(R)>0$. Then the observer can construct the following contract. Whenever $R$ is recommended, advise switching to $L$ and $M$ with equal probability, i.e.\ $\eta_2(R,L)=\eta_2(R,M)=\tfrac12$. Define $S_\eta \in \R^A$ as the aggregate surplus generated when agents unilaterally deviate by following the recommendation $\eta$. 
$$ S_\eta (a) = \sum_i \sum_{\tilde{a}_i \in A_i} \eta_i\left(a_i,\tilde{a}_i\right) u_i\left(\tilde{a}_i, a_{-i}\right) - \sum_i  u_i(a_i,a_{-i})
$$
Consequently, inequality \ref{eq:IC-sum} can be expressed compactly as: 
\begin{equation*}
    \sum_{i} f_i(a) \leq S_\eta (a).
\end{equation*}
\noindent We can then display the surplus as a matrix (see \Cref{tab:game_1_surplus}).  
\begin{table}[h]
\centering
\caption{$S_\eta$ for \Cref{eq:coordination}}\label{tab:game_1_surplus}
\[
\scalebox{1.3}{$
\begin{array}{c|ccc}
  & L & M & R \\
\hline
T & 0 & 0 & \sfrac{1}{2} \\
B & 0 & 0 & \sfrac{1}{2}
\end{array}
$}
\]
\end{table}

By setting the transfer $f(R)=\tfrac12$, with $f(a)=0$ for all $a\neq R$, we generate positive expected profits for the observer while maintaining \Cref{eq:IC-sum}. Thus any strategy profile with $p_i (R)> 0$ is not compatible with a correlated equilibrium.  \hfill $\blacksquare$
\end{example}

\begin{example}\label{ex:miscoordination}[Exploitation of Miscoordination]\\
Sometimes a marginal strategy profile is not compatible with a correlated equilibrium due to the strategies of more than one player. For the game described in \Cref{eq:coordination}, consider the strategy profile $\tilde p$ where $\tilde p_1 (T) = \tilde p_1 (B) = \sfrac{1}{2}$, $\tilde p_2 (L) = \sfrac{1}{4}$ and $\tilde p_2(M) = 3/4$. Intuitively, the probability that $M$ is played seems too high for this strategy profile to be compatible with a correlated equilibrium. 

Indeed, $M$ is an optimal action only if the probability that $B$ is played conditional on $M$ is sufficiently high. To formalize this, consider the following deviation: $\tilde \eta_2 (M, L)  =1$ and $\tilde \eta_i (a_i, a_i') = 0$ whenever $a_i \neq M$ or $a_i^\prime \neq L$. In other words, the observer would recommend to play $L$ with probability 1 whenever the mediator tells player 1 to play $M$.

Once again, we can describe the surplus gained via a table.

\begin{table}[h]
\centering
\caption{$S_{\tilde \eta}$ for \Cref{eq:coordination}}\label{tab:game_1_surplus_2}
\[
\scalebox{1.3}{$
\begin{array}{c|ccc}
  & L & M & R \\
\hline
T & 0 & 9 & 0 \\
B & 0 & -1 & 0
\end{array}
$}
\]
\end{table}

\noindent We see that the deviation proposed implies a tradeoff. A positive surplus is achieved  whenever $(T, M)$ would be played, but a negative surplus is achieved whenever $(B, M)$ would be played. By setting $f_2 (M) =9$, $f_1(B) = -10$ and $f_i(a_i) = 0$ if $a_i \notin \{M, B\}$, the observer gets profits of $7/4$ in expectation, and thus $\tilde{p}$ is not compatible with a correlated equilibrium. \hfill $\blacksquare$

\end{example}

\subsection{Nash Equilibrium} \label{sec: Nash equilibrium}

To provide more context for \Cref{th:correlated_eq}, we extend our line of inquiry to a different solution concept, Nash Equilibrium. Since Nash Equilibria are defined via independent mixtures of marginal strategy profiles, given our data assumption, the analyst is able to directly test whether a marginal strategy profile and a collection of utility functions constitute a Nash Equilibrium. Nonetheless, in this section we develop an analogue of our unexploitability condition from \Cref{sec:correlated equilibrium} for Nash Equilibria in order to understand the content of corelation in the context of the recommendation story. Recall the definition of a Nash Equilibirum.
\begin{definition}
    A marginal strategy profile $p \equiv (p_1,\dots,p_N)$ is a \textbf{Nash Equilibrium} if for all $i \in N$ and all $a_i,a_i' \in A_i$ with $p_i(a_i) >0$, we have the following. 
    \begin{equation*}
        \sum_{a_{-i}}\left(\prod_{j\neq i} p_j(a_{j}) \right) \left[u_i(a_i,a_{-i})-u_i(a_i',a_{-i})\right] \geq 0
    \end{equation*}.
\end{definition}
In words, a strategy profile is a Nash Equilibrium if for each player, every action in the support is weakly better than any other actions in expectation. 
Now, under the assumption that players have quasi-linear utility in money, imagine the point of view of an observer who seeks to derive a profit from the player, as in \Cref{sec:correlated equilibrium}. In particular, she knows $p$ and $u$, and can offer to each player a set of fees $t_i: A \to \mathbb R$ in exchange for a recommendation $\eta_i: A_i \to \Delta (A_i)$. Player $i$ accepts the deal if and only if following the recommendation and accepting the fee constitutes a profitable deviation, that is, for every $a \in A$. 
\begin{equation*}
    u_i(a_i, a_{-i}) \leq \sum_{a_i^\prime } \eta_i(a_i, a_i^\prime) u_i(a_i ^\prime, a_{-i}) - t_i(a)
\end{equation*}

Because of the quasi-linearity assumption, for this equation to hold for each player, there need only exist some ex ante transfer ($h_i$ from \Cref{sec:correlated equilibrium}) between players such that the equation holds when aggregated across players. 
\begin{equation} \label{eq:Nash_sum_IC}
    \sum_{i \in N} u_i(a_i, a_{-i}) \leq \sum_{i \in N} \sum_{a_i^\prime } \eta_i(a_i, a_i^\prime) u_i(a_i ^\prime, a_{-i}) - \sum_{i \in N} t_i(a)
\end{equation}
In a Nash Equilibrium, the joint strategy profile defining play is an independent mixture of the corresponding the joint distribution. As such, Nash Equilibrium play is fully characterized by marginals over actions. As in \Cref{sec:correlated equilibrium}, the observer seeks to make a profit from the fees, but this time, she has access to the joint distribution of actions. Defining $f(a) = \sum_{i \in N} t_i(a)$, the observer seeks to derive a profit in expectation from the trades she makes. 
\begin{equation} \label{eq:Nash_expectation}
    \sum_{a \in A} \left ( \prod_{i} p_i(a_i) \right )f(a) > 0
\end{equation}

We say that a strategy profile is \textbf{unexploitable by profile-wise transfer schemes} if for all possible deviations $\eta_i \in \Delta(A_i)$ and transfer schemes $t_i \in \R^A_i$, both \eqref{eq:Nash_sum_IC} and \eqref{eq:Nash_expectation} cannot be satisfied. 

\begin{theorem} \label{th:Nash}
    A strategy profile is a Nash equilibrium if, and only if, it is unexploitable by profile-wise transfer schemes.
\end{theorem}

\section{Conclusion} \label{sec:conclusion}
In this paper, we provide a dual characterization of correlated equilibrium in terms of exploitability by an outside observer. We show that a marginal strategy profile is compatible with some correlated equilibrium if and only if there exists no system of transfers and action-contingent recommendations that results in the outsider observer making a positive expected gain. This equivalence reframes the incentive constraints defining correlated equilibrium as a joint rationality condition at the level of the group of players.

While not our main motivation, we can also view our characterization as a normative foundation for correlated equilibria. Our main axiom asks that marginal strategy profiles be unexploitable by action-wise transfers. This condition can be interpreted as a notion of group stability. The group of players choose a strategy profile for which no outside observer is incentivized to recommend unilateral deviations. This is reminiscent of \citet{Nau1990} which characterizes which actions can be played in some correlated equilibrium via a no arbitrage condition. Our characterization differs in that we characterize which marginal probabilities (i.e. marginal action profiles) can be played in some correlated equilibrium.

The characterization is useful for applications in which behavior is observed only through empirical distributions of actions. In such settings, checking compatibility with correlated equilibrium can be reframed as testing whether the observed behavior admits a profitable transfer scheme, a problem that naturally takes the form of a linear programming problem. This perspective is particularly useful for detecting collusion or pre-play communication. Empirically, collusion may manifest itself in action distributions that are incompatible with any Nash equilibrium but compatible with a correlated equilibrium. Our characterization provides a tractable test for this distinction, based solely on the exploitability of observed behavior by transfer schemes.

\appendix
\section{Proofs}

\begin{proof}[Proof of \Cref{th:correlated_eq}]
We prove both directions.

\medskip
\noindent
\textbf{($\Rightarrow$)}  
Let $p$ be compatible with a correlated equilibrium $q$. We show that $p$ is unexploitable by action-wise transfer schemes. Suppose, towards a contradiction, that it is exploitable. Then there exist vectors $f_i \in \mathbb{R}^{A_i}$ and kernels $\eta_i \in \Delta(A_i)^{A_i}$ such that, for each $i\in N$,
\begin{equation}
\sum_{a_i\in A_i} p_i(a_i) f_i(a_i) > 0,
\label{eq:exante}
\end{equation}
and for all $a\in A$,
\begin{equation}
\sum_{i\in N}\bigl[u_i(a)+f_i(a_i)\bigr]
\;\le\;
\sum_{i\in N}\sum_{a_i'\in A_i}\eta_i(a_i,a_i')\,u_i(a_i',a_{-i}).
\label{eq:utilitarian}
\end{equation}

Taking expectations of \eqref{eq:utilitarian} with respect to $q$ yields
\begin{align}
\sum_{i\in N}\sum_{a} q(a) f_i(a_i)
&\le
\sum_{i\in N}\sum_{a}\sum_{a_i'} \eta_i(a_i,a_i') q(a)
\bigl[u_i(a_i',a_{-i})-u_i(a)\bigr].
\label{eq:expectation}
\end{align}

Rearranging the right-hand side,
\[
\sum_{i\in N}\sum_{a_i}\sum_{a_i'} \eta_i(a_i,a_i')
\sum_{a_{-i}} q(a_i,a_{-i})
\bigl[u_i(a_i',a_{-i})-u_i(a_i,a_{-i})\bigr].
\]
Since $\eta_i(a_i,a_i')\ge 0$ and $q$ is a correlated equilibrium, each inner sum is nonpositive, hence the entire right-hand side of \eqref{eq:expectation} is weakly negative.

On the left-hand side, using compatibility of $p$ and $q$,
\[
\sum_{i\in N}\sum_{a} q(a) f_i(a_i)
=
\sum_{i\in N}\sum_{a_i} p_i(a_i) f_i(a_i)
> 0,
\]
by \eqref{eq:exante}. This contradicts \eqref{eq:expectation}. Therefore, $p$ is unexploitable.

\medskip
\noindent
\textbf{($\Leftarrow$)}  
Suppose $p$ is not compatible with any correlated equilibrium. Then there exists no function $q:A\to\mathbb{R}_+$ satisfying:
\begin{itemize}
\item for all $i\in N$ and $a_i,a_i'\in A_i$,
\[
\sum_{a_{-i}} q(a_i,a_{-i})
\bigl[u_i(a_i,a_{-i})-u_i(a_i',a_{-i})\bigr]\ge 0,
\]
\item for all $i\in N$ and $a_i\in A_i$,
\[
p_i(a_i)=\sum_{a_{-i}} q(a_i,a_{-i}),
\]
\item for all $a\in A$, $q(a)\ge 0$.
\end{itemize}

By Farkas' Lemma, the infeasibility of this system implies the existence of vectors
$f_i\in\mathbb{R}^{A_i}$ and nonnegative coefficients $\eta_i(a_i,a_i')\ge 0$ such that
\begin{equation}
\sum_{i\in N}\sum_{a_i\in A_i} p_i(a_i) f_i(a_i) > 0,
\label{eq:expected}
\end{equation}
and for all $a\in A$,
\begin{equation}
\sum_{i\in N}\sum_{a_i'\in A_i}
\eta_i(a_i,a_i')
\bigl[u_i(a_i,a_{-i})-u_i(a_i',a_{-i})\bigr]
+
\sum_{i\in N} f_i(a_i)
\le 0.
\label{eq:dual}
\end{equation}

Rewriting \eqref{eq:dual},
\[
\sum_{i\in N} u_i(a)
\;\le\;
\sum_{i\in N}\sum_{a_i'\in A_i}
\eta_i(a_i,a_i')\,u_i(a_i',a_{-i})
-
\sum_{i\in N} f_i(a_i).
\]

Since the dual inequalities are homogeneous, we may rescale the coefficients without loss of generality. In particular, we may normalize so that
$\sum_{a_i'} \eta_i(a_i,a_i') = 1$ for each $i$ and $a_i$. Then each $\eta_i(a_i,\cdot)$ lies in $\Delta(A_i)$, and the above inequality shows that $(f_i,\eta_i)$ constitute an action-wise transfer scheme that exploits $p$, together with \eqref{eq:expected}.

Hence $p$ is exploitable, completing the proof.
\end{proof}

\begin{proof}[Proof of \Cref{th:Nash}]
Fix a mixed strategy profile $(p_i)_{i\in N}$ and let $q\in\R^{A}$ denote the induced joint distribution,
\[
q(a)\coloneqq \prod_{i\in N} p_i(a_i)\qquad\text{for all }a\in A.
\]
Then $(p_i)_{i\in N}$ is a Nash equilibrium if and only if $q$ satisfies, for all $i\in N$ and all $a_i,a_i'\in A_i$,
\begin{equation}\label{eq:nash_ineq_q}
\sum_{a_{-i}} q(a_i,a_{-i})\bigl[u_i(a_i,a_{-i})-u_i(a_i',a_{-i})\bigr]\ge 0.
\end{equation}

Consider the following feasibility problem in the variables $\tilde q\in\R^{A}$:
\begin{itemize}
    \item[(i)] For all $i\in N$ and all $a_i,a_i'\in A_i$,
    \[
    \sum_{a_{-i}} \tilde q(a_i,a_{-i})\bigl[u_i(a_i,a_{-i})-u_i(a_i',a_{-i})\bigr]\ge 0.
    \]
    \item[(ii)] For all $a\in A$, $\tilde q(a)=q(a)$.
    \item[(iii)] For all $a\in A$, $\tilde q(a)\ge 0$.
\end{itemize}
Since (ii) pins down $\tilde q$ uniquely, the system is feasible if and only if the fixed vector $q$ satisfies \eqref{eq:nash_ineq_q}, i.e., if and only if $(p_i)_{i\in N}$ is a Nash equilibrium.

Apply Farkas's Alternative (Theorem 38 in \citet{borderAlternative}) to the feasibility system (i)--(iii), after moving the left-hand side of each inequality in (i) to the right-hand side. We obtain that (i)--(iii) is feasible if and only if there do \emph{not} exist multipliers $\eta_i(a_i,a_i')\ge 0$ for all $i\in N$ and $a_i,a_i'\in A_i$, and numbers $f(a)\in\R$ for all $a\in A$, such that
\begin{align}
\label{eq:dual_profit}
0> & \sum_{a\in A} f(a)\,q(a)\,,\\
\label{eq:dual_domination}
f(a) &\ge \sum_{i\in N}\sum_{a_i'\in A_i}\eta_i(a_i,a_i')\bigl[u_i(a)-u_i(a_i',a_{-i})\bigr]
\qquad\text{for all }a\in A.
\end{align}

We now justify the normalization of the multipliers. Suppose $(\eta,f)$ satisfies \eqref{eq:dual_profit}--\eqref{eq:dual_domination}. Then for any scalar $t>0$, the pair $(t\eta,tf)$ also satisfies \eqref{eq:dual_profit}--\eqref{eq:dual_domination}. Indeed, \eqref{eq:dual_domination} is homogeneous of degree one in $(\eta,f)$, and \eqref{eq:dual_profit} becomes $t\sum_a f(a)q(a)<0$, which remains true for all $t>0$. Hence we are free to rescale $(\eta,f)$ by an arbitrary positive constant.

Using this freedom, choose $t>0$ so that
\[
\sum_{a_i'\neq a_i} (t\,\eta_i(a_i,a_i'))\le 1
\quad\text{for all } i\in N,\ a_i\in A_i.
\]
(For instance, if $M\coloneqq \max_{i,a_i}\sum_{a_i'\neq a_i}\eta_i(a_i,a_i')>0$, take $t=1/M$, and if $M=0$ there is nothing to do.)
Then define
\[
\eta_i(a_i,a_i)\;=\;1-\sum_{a_i'\neq a_i}\eta_i(a_i,a_i').
\]
This does not affect \eqref{eq:dual_domination}, since the coefficient on $\eta_i(a_i,a_i)$ is
$u_i(a)-u_i(a_i,a_{-i})=0$ for every $a\in A$. Under this convention, for each $(i,a_i)$ the vector
$\eta_i(a_i,\cdot)$ is a probability distribution over $A_i$, and \eqref{eq:dual_domination} is equivalent to
\begin{equation}\label{eq:dual_domination_norm}
f(a)\;\ge\;\sum_{i\in N}\left[u_i(a)-\sum_{a_i'\in A_i}\eta_i(a_i,a_i')\,u_i(a_i',a_{-i})\right]
\qquad\text{for all }a\in A.
\end{equation}

By construction, the existence of $(\eta,f)$ satisfying \eqref{eq:dual_profit}--\eqref{eq:dual_domination} (equivalently, \eqref{eq:dual_profit} and \eqref{eq:dual_domination_norm}) is exactly the negation of feasibility of (i)--(iii). Therefore, the system (i)--(iii) is feasible if and only if there is no such $(\eta,f)$. By multiplying \eqref{eq:dual_profit} and \eqref{eq:dual_domination_norm} by negative one, we recover the equations defining unexploitable by profile-wise transfer schemes. In conclusion, $(p_i)_{i\in N}$ is a Nash equilibrium if and only if it is unexploitable by profile-wise transfer schemes.
\end{proof}

\bibliographystyle{ecta}
\bibliography{CorM}

\end{document}